\documentclass[11pt,letterpaper]{article}
\usepackage{fullpage}
\usepackage{lmodern}
\usepackage[T1]{fontenc}
\usepackage[utf8]{inputenc}
\usepackage[tbtags]{amsmath}
\usepackage{textcomp}
\allowdisplaybreaks
\usepackage{amssymb,amsthm}
\usepackage{hyperref}
\usepackage[svgnames]{xcolor}
\usepackage[capitalise,nameinlink]{cleveref}
\definecolor{links}{RGB}{11, 85, 255}
\definecolor{cites}{RGB}{0, 200, 0}
\definecolor{urls}{RGB}{255, 116, 0}
\hypersetup{colorlinks={true},linkcolor={links},citecolor=[named]{cites},urlcolor={urls}}
\usepackage[numbers,compress]{natbib}
\usepackage{doi}
\usepackage{hyphenat}
\usepackage[font=footnotesize]{caption}

\usepackage{tikz}  %needed for \textcolor as well as tikz
\usetikzlibrary{automata, positioning, arrows}
\usetikzlibrary{shapes}
\usetikzlibrary{shapes.geometric}
\usetikzlibrary{arrows}
\usetikzlibrary{chains,shapes.multipart}
\usetikzlibrary{shapes,calc}
\usetikzlibrary{matrix,backgrounds, arrows.meta}
\pgfdeclarelayer{bg}    % declare background layer
\pgfsetlayers{bg,main} % set the order of the layers (main is the standard layer)
\usepackage{import}

\usepackage{thm-restate}
\usepackage{cleveref}
\usepackage{subcaption}

\usepackage[ruled,vlined,linesnumbered]{algorithm2e}

\SetKwFor{WP}{With probability}{}{:}{}
\SetKwBlock{Prob}{With probability}{end}
\makeatletter
\newcommand*{\algotitle}[2]{%
	\stepcounter{algocf}%
	\hypertarget{algocf.title.\theHalgocf}{}%
	\NR@gettitle{#1}%
	\label{#2}%
	\addtocounter{algocf}{-1}%
}

\newtheorem{theorem}{Theorem}

\theoremstyle{definition}

\author{
	Philip Lazos\footnote{Sapienza University of Rome: 
		\url{plazos@gmail.com}}
	\and
	Francisco J. Marmolejo-Coss\'io\footnote{University of Oxford, IOHK: 
		\url{francisco.marmolejo@cs.ox.ac.uk}}
	\and
	Xinyu Zhou\footnote{University of Maryland: 
		\url{xyzhou@terpmail.umd.edu} and \url{jkatz2@gmail.com}}
	\and
	Jonathan Katz{\textsuperscript{\textsection}}
}

\title{RPPLNS: Pay-per-last-N-shares with a Randomised Twist\thanks{Philip 
Lazos is supported by the ERC Advanced 
Grant 788893 AMDROMA ``Algorithmic and Mechanism Design Research in 
Online Markets'' and the MIUR PRIN project ALGADIMAR ``Algorithms, Games, 
and Digital Markets.''}}

\date{\today}

\begin{document}
\maketitle

\begin{abstract}
``Pay-per-last-$N$-shares'' (PPLNS) is one of the most common payout 
strategies used by mining pools in Proof-of-Work (PoW) cryptocurrencies. As 
with any payment scheme, it is imperative to study issues of incentive 
compatibility of miners within the pool. For PPLNS this question has only been 
partially answered; we know that reasonably-sized miners within a PPLNS pool 
prefer following the pool protocol over employing {\em specific} deviations. In 
this paper, we present a novel modification to PPLNS where we randomise the 
protocol in a natural way. We call our protocol ``Randomised 
pay-per-last-$N$-shares'' (RPPLNS), and note that the randomised structure of 
the protocol greatly simplifies the study of its incentive compatibility. We show 
that RPPLNS maintains the strengths of PPLNS (i.e., fairness, variance 
reduction, and resistance to pool hopping), while also being robust against a 
richer class of strategic mining than what has been shown for PPLNS.  
\end{abstract}

\section{Introduction}

In Bitcoin, miners maintain a ledger of transactions and are rewarded for their 
efforts by the underlying protocol. Successfully appending a block to the ledger 
is computationally difficult; it can take common computing equipment years to 
find a single block in expectation. In response to this variability, miners often 
pool their resources so that rather than rarely earning large rewards they earn 
smaller rewards at a more consistent rate. 

In more detail, we can think of the Bitcoin ecosystem as consisting of $n$ 
strategic miners, $m_1$,...,$m_n$, with hash powers $\alpha_1,...,\alpha_n > 0$, 
such that $\sum_i \alpha_i = 1$. Intuitively, $\alpha_i$ represents the proportional 
computational power that a miner has. Assuming that block mining happens in 
discrete time-steps, the probability that $m_i$ mines the next block is equal to 
$\alpha_i$. The dilemma of the previous section corresponds to a single miner 
having a small $\alpha_i$ value, and hence only mining a block in $1/\alpha_i$ 
time steps in expectation. On the other hand, a set of $S$ miners could combine 
their computational power and have a $\sum_{m_i \in S} \alpha_i$ chance of 
mining the next block, sharing the rewards if they manage to do so. 

In order to share rewards, the pool must have a way of identifying the 
computational contribution of each miner. This is done by accepting partial 
proofs of work, which are ``near-misses'' to Bitcoin's desired hash rate and 
whose number is directly proportional to the computational effort spent on 
extending the blockchain. The pool operator collects these ``near-misses'', 
called shares, reported by every miner in the pool, and uses them to distribute 
payments once an actual block is found. Ideally, participating in a pool should 
have (at least) the following guarantees:
\begin{enumerate}
	\item Fairness: miners earn the same block reward in expectation as mining 
	alone.
	\item Variance reduction: for any fixed amount of time, miners have lower 
	variance in block reward than when mining alone.
	\item Robustness against pool hopping: at no point of time is there is a 
	benefit in leaving the pool for another one or leaving the pool to mine 
	individually.
	\item Incentive Compatibility: to maximize their reward, each participating 
	miner should always expend maximum effort and report shares/blocks 
	immediately as they are generated.
\end{enumerate}

One of the most popular pool mining protocols which (partially) satisfies these 
properties is ``Pay-per-last-$N$-shares'' (PPLNS). Miners report shares to a 
pool operator which maintains a queue of the $N$ most recent shares reported 
to it, and if a block is found and reported by the pool, the owners of these $N$ 
shares are paid proportionally ($1/N$ times the value of a block for each such 
share). The structure of PPLNS is such that it satisfies properties 1-3 above 
\cite{rosenfeld2011analysis}. With respect to property 4, 
\cite{zolotavkin2017incentive, schrijvers2016incentive} demonstrate that miners in 
a PPLNS are incentivised to act honestly if they are only permitted specific 
deviations, hence PPLNS only partially satisfies property 4. 

\paragraph{Our Contributions.}
In this paper, we present a novel modification to PPLNS where we randomise 
the protocol in a natural way. We call our protocol ``Randomised 
pay-per-last-$N$-shares'' (RPPLNS), and note that the randomised structure of 
the protocol greatly simplifies the study of its incentive compatibility. We show 
that RPPLNS maintains the strengths of PPLNS (i.e., fairness, variance 
reduction, and resistance to pool hopping), while also being robust against a 
richer class of strategic mining than what has been shown for PPLNS. In 
particular, \citet{schrijvers2016incentive} showed that PPLNS is robust when 
strategies are limited to either mining honestly \emph{or} secretly hoarding all 
accumulated shares and suddenly releasing them if a block is found. For 
RPPLNS  we provide experimental evidence that suggests that mining honestly 
maximizes rewards among any strategy involving hoarding and releasing shares.

\subsection{Structure of the Paper}
The structure of the paper is as follows: in Section \ref{sec:relatedwork} we 
present related literature, and in Section \ref{sec:state-machine-pools} we define 
a formal framework for pool mining protocols that encompasses PPLNS and 
RPPLNS in a language that is amenable to a rigorous mathematical analysis of 
their properties. Section \ref{sec:RPPLNS} holds the key results of our paper 
where we describe RPPLNS and prove its desirable properties. In Section 
\ref{sec:strategic-hoarding} we give a theoretical justification for why specific 
(improbable) initial pool states give rise to strategic mining in both PPLNS and 
RPPLNS. Finally, in Section \ref{sec:future-work} we describe possible further 
areas of research. 

\subsection{Related Work}\label{sec:relatedwork}
Our work follows a line of research in the game theoretic aspects of 
cryptocurrencies. Strategic mining has been studied from the inception of 
Bitcoin when Nakamoto suggested the robustness of Bitcoin to double-spend 
attacks \cite{nakamoto2008bitcoin}. Subsequently, this area of research has 
taken multiple directions. In \cite{eyal2014majority}, the authors demonstrate that 
honest mining is not robust to strategic mining in terms of block reward, even 
when a miner has less than a majority computational stake in the Bitcoin 
ecosystem, by exhibiting a specific mining strategy, selfish mining. This work is 
refined in \cite{sapirshtein2016optimal}, \cite{nayak2016stubborn} and 
\cite{kiayias2016blockchain} by generalising selfish mining, pairing selfish mining 
with network-level attacks and proving limited incentive compatibility of honest 
mining if miners have low enough hash power. Further incentives at the 
individual miner level have also been studied in \cite{carlsten2016instability} 
where transaction rewards alone are demonstrated to be unstable for incentive 
compatibility of mining. At the pool level, the author of \cite{eyal2015miner} 
demonstrates an infiltration attack pools can wage against each other that leads 
to an iterated prisoner's dilemma between pool operators (dubbed the miner's 
dilemma). This has been further refined in \cite{KwonKSVK17} where a judicious 
refinement of an infiltration attack can avoid the miner's dilemma, so that larger 
pools benefit from attacking at a cost to smaller pools. 

This aforementioned direction of research is mainly involved with strategic 
mining with respect to block rewards, but our work is dedicated to strategic 
mining within pools with specified payment protocols (see 
\cite{rosenfeld2011analysis} for an extensive survey of pool protocols). In 
\cite{schrijvers2016incentive}, the authors study incentive compatibility in pool 
protocols that decide how to make payments on the basis of the quantity of 
shares each miner reported, irrespective of the order in which they are received. 
In this setting, they give necessary conditions for when behaving honestly 
dominates strategic hoarding of shares and blocks for pool miners. Not only do 
they show that certain popular strategies are not incentive compatible or 
budget-balanced (proportional payment, fixed payment per share), but they also 
provide an incentive-compatible pool protocol that is budget-balanced in this 
setting where the order of shares are irrelevant to the pool operator. A PPLNS 
pool fundamentally needs the order in which shares are submitted to decide how 
much to pay miners, hence the previous results do not hold for PPLNS. For this 
reason, the authors in \cite{schrijvers2016incentive} also partially answer the 
question of whether PPLNS is an incentive compatible mining pool scheme, by 
showing that at small enough hash powers and large enough share difficulty 
levels, honest mining is robust against the specific strategic deviation in which 
miners keep a single share private with the hopes of finding a subsequent block 
to ensure said private share is paid, by publishing it immediately before 
revealing a block. Subsequently, the authors of \cite{zolotavkin2017incentive} 
study a different class of strategic deviations in PPLNS where a miner hoards a 
certain number $x \in \mathbb{N}$ of shares. Subsequent shares are published 
immediately, and whenever a block is found, those $x$ shares are published 
immediately before publishing the block. Their analysis makes the assumption 
that each strategic miner reaches their threshold $x$, and show when being 
honest outperforms being strategic in this setting. Finally, the authors of 
\cite{qin2019novel} exhibit specific reporting strategies that can be beneficial to 
strategic miners at high enough hash rates.

\section{PPLNS and RPPLNS}
\label{sec:state-machine-pools}

In this section we give a framework for studying miner incentives in PPLNS and 
RPPLNS mining pools. To do so, we begin with a high-level overview with key 
aspects of pool protocols.

\subsection{Pool Basics}
In general, the pool operator prepares the `template' of the next block in 
advance to send to its pool miners: this contains all the information about where 
the next block should be added and what it should contain. Then, pool miners 
try to fill in different nonce values at an attempt to obtain a valid block (i.e., one 
with a low enough hash value) whose rewards will be claimed by the pool and 
redistributed to its miners. When mining solo (rather than as a part of a pool), 
miners obtain an expected reward proportional to their hash power. In order to 
fairly redistribute rewards, a pool needs a proxy to measure the computational 
resources a miner is contributing to the pool's operations. This is done by 
allowing pool miners to send the pool operator near-misses (also called 
\emph{shares} in this context) to the Bitcoin difficulty threshold. In our full 
exposition of mining pool rules, we assume the following: that there is a mining 
pool composed of $k$ miners, which will call pool miners, and that all miners 
outside the pool (non-pool miners) are honest, and hence can be modelled as a 
single honest miner. The reason we model the other miners as a single honest 
entity is due to the fact that we focus on strategic mining deviations {\em within} 
the pool protocol only. Pool miners are denoted by $m_1,...,m_k$, and the 
non-pool miner by $m_0$. As before, the $i$-th miner has hash power $\alpha_i$.

\newcommand{\honest}{\ensuremath{m_2}}
\newcommand{\strategic}{\ensuremath{m_1}}
\newcommand{\strategicR}{\ensuremath{R_1}}
\newcommand{\nonpool}{\ensuremath{m_0}}

\paragraph{Valid Share / Block Generation.} The first relevant parameter that 
the pool must set is the {\em relative difficulty} of blocks to shares, which can be 
parametrised by a non-negative number $D \geq 1$. We model share and block 
generation in discrete time-steps. Each turn a share is created, and it is 
attributed to miner $m_i$ with probability $\alpha_i$. Furthermore, each share 
has a $1/D$ probability of being a block. A block can be thought of as a share 
that gives the pool a reward to be distributed amongst its miners.

\paragraph{Messages to the Pool.} At any given time-step, the pool receives 
messages from miners which can be of the form ${\sf share}_i$ or ${\sf block}_i$ 
if $m_i$ reports a share or block respectively to the pool\footnote{$\strategic$ 
only sends block messages to the pool due to the fact that the concept of a 
share only makes sense for pool miners. On the other hand, a block message 
from $\strategic$ corresponds to the pool becoming aware of said block on the 
global blockchain itself.}. We denote the set of messages the pool operator can 
receive by $M$. 

\paragraph{Honest and Strategic Pool Mining.} For all pool mining schemes, 
we say that a pool miner is honest if they mine upon the block given by the pool 
operator and if they report shares/blocks immediately. For independent miners, 
we say they are honest if they publish blocks they find immediately (The concept 
of a share is irrelevant outside the pool). Consequently, strategic miners have 
the possibility of arbitrarily hoarding shares and blocks (valid messages of the 
form ${\sf share}_i$ or ${\sf block}_i$) to be published at a time after they are 
found. It is crucial to note that the mechanics of pool mining are such that these 
deviations do not come without a potential cost. Pool miners are all given a 
block template to work on, and shares are near-misses to the hash of a block 
within this template. If a miner is hoarding shares/blocks and another honest 
pool miner publishes a block to the pool, these shares/blocks are no longer 
valid and the pool will not pay the owner of them. The same scenario happens if 
an honest miner outside the pool mines a block. Strategic miners must therefore 
balance this risk while hoarding. Furthermore, if hoarding is beneficial for some 
miners, the pool's overall quality degrades in two ways: for them to increase 
their payoff some other miners will have to earn less and there is potential for 
decreased computational efficiency as some work might be performed twice.

\paragraph{Payments.} Pool operators maintain a history of messages 
received and as a function of this history decide how to redistribute block 
rewards to pool miners. Notice that since the pool sets the block template for 
pool miners, only the pool operator can receive and redistribute funds from a 
valid block found by the pool. Setting a pool payment scheme determines miner 
incentives, hence it is crucial to analyse whether mining honestly is beneficial to 
pool miners. In general, pools are not obliged to redistribute the entirety of their 
block reward, but in this paper we focus on pools which do. We call these pools 
{\em budget balanced}.

\subsection{PPLNS}

At a high level, PPLNS pools are simple to understand. A pool operator sets a 
relative block to share difficulty, $D$, and a queue size parameter, $N \in 
\mathbb{N}$. Pool miners report shares and blocks to the pool operator, which 
in turn maintains a queue of length $N$ consisting of the the most recent shares 
reported. If ever a block is found, the pool operator pays the shares in the queue 
proportionally. This process can be visualised in Figure \ref{fig:PPLNS}. 

Since we are studying strategic mining deviations, we have to precisely define 
the pool protocol, as this will govern whether miners are incentivised to be 
honest in the first place. We do this by modelling PPLNS as a deterministic 
state machine\footnote{Many common pool protocols can be cast in this 
framework. This is explored further in \cite{marmolejo2019equilibrium}}. At any 
given moment, the pool operator maintains a state $s \in S = (\{*\} \cup [n])^N$. 
$S$ represents the state space of the PPLNS queue (all possible queue 
compositions), where ``$*$'' represents an empty value while the queue is filled 
in the first $N$ messages the pool operator receives. Furthermore, we let $s_0 = 
(*)_{i=1}^N$ be a distinguished initial state of the pool protocol that corresponds 
to an empty queue. 

As with any state machine, we need to define a transition function between 
states. We recall that $M$ denotes the set of messages that a pool operator can 
receive, hence we can specify a transition function $T: S \times M \rightarrow 
S$. For a given message from the $i$-th pool miner (meaning $i \neq 0$), $x = 
{\sf share}_i$ or $x = {\sf block}_i$, $T(s,x) = i:s_{<n}$, which is the 
concatenation of $i$ with the first $n-1$ elements of $s$ (The queue is filled from 
the left to the right). As for non-pool miner messages, $T(s,{\sf block}_0) = s$.

Finally, we need to specify a payment protocol in the scenario where a pool 
miner finds a block. This means that the pool receives a message ${\sf block}_i$ 
for $i \neq 0$, which in turn implies that if the pool is in state $s$, it transitions to 
state $s' = T(s,{\sf block}_i)$. State $s'$ contains the identities of miners with the 
most recent $N$ shares, hence we let $P:S\rightarrow \mathbb{R}^k$ be a 
payment function such that $P(s')_i = \frac{|\{j \ \mid \ s'_{j} = i\}| }{|\{j \ \mid \ 
s'_{j} \neq * \}|} \geq 0$ is the payment $m_i$ receives from the found block.

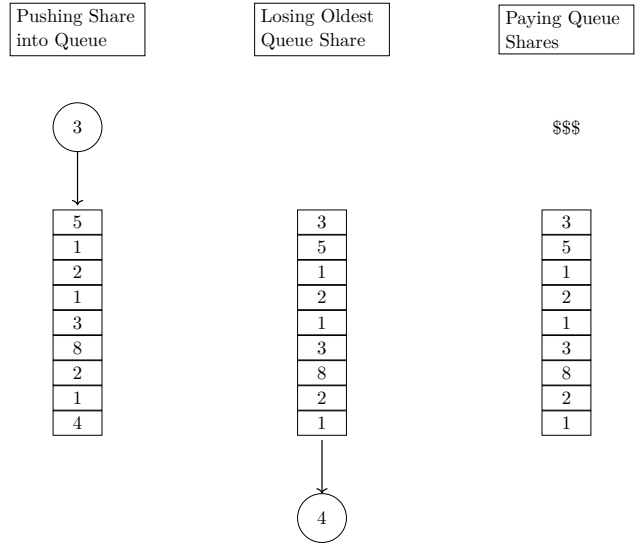
\begin{figure*}
	\centering
	\scalebox{0.65}{
		\begin{tikzpicture}[draw, minimum width=1cm, minimum height=0.5cm]
		
		\node[draw,text width=2.5cm] at (-5,6) {Pushing Share into Queue};
		
		\node [draw,circle] at (-5,4) {3};
		
		\matrix (queue) at (-5,0) [matrix of nodes, nodes={draw, nodes={draw}}, 
		nodes in empty cells]
		{
			5 \\ 1 \\ 2 \\ 1 \\ 3 \\ 8 \\ 2 \\ 1 \\ 4 \\
		};
		
		\draw[->, line width=0.25mm] (-5,3.5) -- (-5,2.4);    
		
		\node[draw,text width=2.5cm] at (0,6) {Losing Oldest Queue Share};
		
		\matrix (queue) at (0,0) [matrix of nodes, nodes={draw, nodes={draw}}, 
		nodes in empty cells]
		{
			3 \\ 5 \\ 1 \\ 2 \\ 1 \\ 3 \\ 8 \\ 2 \\ 1 \\
		};
		\draw[->, line width=0.25mm] (0,-2.4) -- (0,-3.5);    
		\node [draw,circle] at (0,-4) {4};
		
		\node[draw,text width=2.5cm] at (5,6) {Paying Queue Shares};
		\node at (5,4) {\$\$\$};
		
		\matrix (queue) at (5,0) [matrix of nodes, nodes={draw, nodes={draw}}, 
		nodes in empty cells]
		{
			3 \\ 5 \\ 1 \\ 2 \\ 1 \\ 3 \\ 8 \\ 2 \\ 1 \\
		};

		\end{tikzpicture}
	}
	
	\caption{A single transition of state for PPLNS. Share 3 is pushed into the 
	queue, causing share 4 at the end of the queue to exit. Upon this transition, 
	all owners of shares in the queue are paid $1/N$ per share in the queue. }
	\label{fig:PPLNS}
\end{figure*}

\subsection{RPPLNS}

RPPLNS is also simple to understand at a high level. The protocol is similar to 
PPLNS with the added difference that shares are no longer maintained in a 
queue but rather a bag that does away with the sequential nature of when which 
share arrived. The relevant parameters for a pool are still the relative difficulty of 
blocks to shares, $D$, and, $N \in \mathbb{N}$, the size of the bag of shares 
maintained by the pool operator. Pool miners report shares and blocks to the 
pool operator. If the bag of shares maintained by the operator is not full (i.e. 
there are less than $N$ shares in the bag), then the reported share is 
automatically added to the bag. On the other hand, when the bag is full, the 
reported share displaces a random share in the bag maintained by the operator. 
If a block is found, first it is added to the bag as a share via the aforementioned 
method, and subsequently the pool operator pays the shares in the bag with a 
proportional value of the block reward. This process is visualised in Figure 
\ref{fig:RPPLNS}.

Once more, we must be rigorous to study incentive compatibility of pool miners. 
With this goal in mind, we model RPPLNS by using a randomised transition 
state machine. At any moment of time, the pool operator maintains a state $s \in 
S = \{s \in [N]^k \mid \sum s_i \leq N \}$. For a given state, $s$, we let $s_i$ 
represent the number of shares that $m_i$ owns in the protocol bag. We also let 
$s_0 = \vec{0} \in [N]^k$ be a distinguished initial state of the pool protocol 
which corresponds to an empty bag.

We now define the randomised transitions that an RPPLNS pool takes upon 
receiving messages from miners. In order to do so, we let $\Delta(S)$ denote the 
set of all probability distributions over $S$. We let $T:S \times M \rightarrow 
\Delta(S)$ be the randomised transition function of RPPLNS. This means that if 
the pool is in state $s \in S$ and receives message $x \in M$, then the resulting 
state $s'$ will be distributed according to $T(s,x)$, which we define as follows:
\begin{itemize}
	\item If $x = {\sf block}_0$, $\mathbb{P}_{T(s,x)}( s') = \mathbb{I}(s' = s)$, the 
	indicator function for $s$. In other words, the pool does not change state.
	\item If $x \in \{{\sf share}_i, {\sf block}_i\}$ such that $i \neq 0$, and in 
	addition $\sum s_{i=1}^k < N$, then $\mathbb{P}_{T(s,x)}( s') = \mathbb{I}(s' = 
	s + e_i)$. In other words, the pool's bag is not full in $s$, hence it 
	deterministically adds $m_i$'s share to the bag.
	\item If $x \in \{{\sf share}_i, {\sf block}_i\}$ such that $i \neq 0$, and in 
	addition $\sum s_{i=1}^k = N$, then $\mathbb{P}_{T(s,x)} \left(s - e_j + e_i 
	\right) = \frac{s_j}{N}$. In other words, the pool operator picks a random share 
	from $s$ to kick out to make way for $m_i$'s newly reported share.
\end{itemize}

Finally, we define the payment that occurs if a block is found. Suppose that 
upon receiving a block, the pool state transitions (randomly) to state $s'$. As 
with PPLNS, we let $P:S\rightarrow \mathbb{R}^k$ be a payment function such 
that $P(s')_i = \frac{s_i}{\sum_{j=1}^k s_j} \geq 0$ is the payment $m_i$ receives 
from the found block.

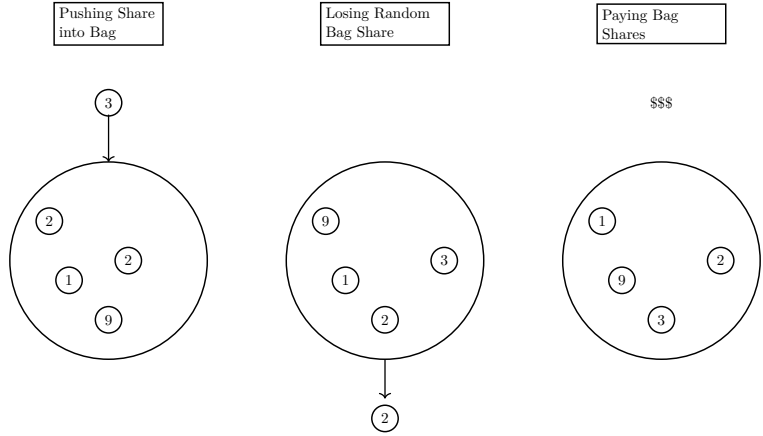
\begin{figure*}
	\centering
	\scalebox{0.7}{
		\begin{tikzpicture}[thick,scale=0.75, every node/.style={transform shape}]
		%\begin{tikzpicture}[draw, minimum width=1cm, minimum height=0.5cm]
		
		\node[draw,text width=2.5cm] at (-5,2) {Pushing Share into Bag};
		
		\node at (-5,0) [circle, draw] {3};
		\draw (-5,-4) circle [radius=2.5cm];
		\node at (-5,-5.5) [circle, draw] {9};
		\node at (-4.5,-4) [circle, draw] {2};
		\node at (-6.5,-3) [circle, draw] {2};
		\node at (-6,-4.5) [circle, draw] {1};
		\draw[->, line width=0.25mm] (-5,-0.3) -- (-5,-1.5);
		
		\node[draw,text width=3cm] at (2,2) {Losing Random Bag Share};
		
		\draw (2,-4) circle [radius=2.5cm];
		\node at (2,-5.5) [circle, draw] {2};
		\node at (3.5,-4) [circle, draw] {3};
		\node at (0.5,-3) [circle, draw] {9};
		\node at (1,-4.5) [circle, draw] {1};
		\node at (2,-8) [circle, draw] {2};
		\draw[->, line width=0.25mm] (2,-6.5) -- (2,-7.5);   
		
		\node[draw,text width=3cm] at (9,2) {Paying Bag Shares};
		\node at (9,0) {\$\$\$};
		
		\draw (9,-4) circle [radius=2.5cm];
		\node at (9,-5.5) [circle, draw] {3};
		\node at (10.5,-4) [circle, draw] {2};
		\node at (7.5,-3) [circle, draw] {1};
		\node at (8,-4.5) [circle, draw] {9};
		
		\end{tikzpicture}
	}
	
	\caption{A single turn's randomised transition for RPPLNS. The pool receives 
	share 3 and pushes it into the bag by randomly selecting an existing share in 
	the bag to kick out. In this case, share 2 is picked to leave. Upon the 
	randomised push, all owners of shares in the bag are paid $1/N$ per share in 
	the bag.}
	\label{fig:RPPLNS}
\end{figure*}

\section{Properties of RPPLNS}
\label{sec:RPPLNS}

In the analysis of RPPLNS, it suffices to consider a single strategic pool miner, 
$\strategic$ with hash power $\alpha$, a single honest pool miner $\honest$ with 
hash power $\beta$, and a single honest non-pool miner $\nonpool$ with hash 
power $\gamma$. Indeed, $\honest$ and $\nonpool$ could be composed of 
multiple honest miners, but if they are honest, we can model their behaviour as 
that of a single miner of their aggregate hash power. In addition, to model 
revenues, we consider a turn-based process. Every turn, either $\strategic$, 
$\honest$ or $\nonpool$ find a share with probability $\alpha$, $\beta$ and 
$\gamma$ respectively, and each share has a further $\frac{1}{D}$ probability of 
being a full block. We wish to point out that we say that $\nonpool$ finds shares 
in the sense that it computes a block with a hash that is a near-miss to the target 
hash (by a factor of $D$). $\nonpool$ does not actually report this near miss to 
the pool since it is not a part of the pool. However, $\nonpool$ does publish 
blocks immediately to all agents of the Bitcoin ecosystem, so we consider this 
as a message to the pool operator. Finally, since $\honest$ is an honest pool 
miner, whenever they find a share/block they communicate this immediately to 
the RPPLNS pool. 

\subsection{Fairness and Variance Reduction}
\label{sec:fairness-var}

We show that if $\strategic$ is honest, then their expected block reward per turn 
is precisely $\alpha/D$. Since each share has a $\frac{1}{D}$ probability of being 
a block, this coincides with the expected $\alpha$ block reward $\strategic$ 
would get (per block mined by the system) mining solo. In addition, we 
demonstrate that RPPLNS enjoys similar variance reduction in block reward to 
what characterises PPLNS.

%\begin{theorem}
\begin{restatable}{theorem}{variance}
	\label{thm:RPPLNS-fair-variance}
	Suppose that $\strategic$ is honest with hash power $\alpha$, then their 
	expected per-turn block reward is $\frac{\alpha}{D}$ in an RPPLNS mining 
	pool. In addition, the variance of their per-turn block reward is $ 
	\frac{1}{D^2}(\alpha - \alpha^2) + \frac{\alpha}{ND}$
\end{restatable}

\begin{proof}
	
	Suppose that $\strategic$ finds a share and sends it to the pool manager. Let 
	$X$ be the random variable that specifies how much revenue that share 
	makes over its lifetime in the bag, $Z$. Since RPPLNS is a push-pay 
	protocol, this means that $Z-1$ is a geometric random variable with mean 
	$N-1$, since the first turn a share is sent to the protocol it is automatically 
	added to the bag and thus elligible for payment, whereas once the share is 
	kicked out of the bag it is not elligible for the subsequent turn's payment. At 
	each turn over the lifetime of the the share in the bag, let $Y_i$ be the 
	indicator random variable for whether any miner in the pool (i.e. $\strategic$ or 
	$\honest$) finds a full solution at the $i$-th turn (including the initial turn when 
	the share is found, as a share can be a full solution as well). This means that 
	the revenue of the share at the $i$-th turn is $Y_i/N$.
	$$
	X = \sum_{i=1}^Z (Y_i/N)
	$$
	Clearly each of the $Y_i$ is iid, hence we can use Wald's equation to get
	$$
	\mathbb{E}(X) = \mathbb{E}(Z) \mathbb{E}(Y_i) = (N)(1/ND) = 1/D
	$$
	With this in hand, we know that the expected revenue of an honest miner in 
	any given turn is the previous expression multiplied by the probability of 
	getting a share (including a share that is a full solution). Thus the expected 
	revenue of the $\strategic$ is 
	$$
	\mathbb{E}(\strategicR) = \alpha/D
	$$
	
	As for the second part of the theorem, we wish to compute the variance of 
	$X$. To do so, we compute $\mathbb{E}(X^2)$:
	$$
	\mathbb{E}(X^2 \mid Z) = \frac{1}{N^2} \mathbb{E} \left( \sum_{i=1}^Z Y_i^2  
	\mid Z\right)
	$$
	The inner sum can be expressed as 
	$$
	\sum_{i<j} \mathbb{E}(Y_i Y_j) + \sum_{i=1}^Z \mathbb{E}(Y_i)
	$$
	Since these are iid we get
	$$
	\binom{Z}{2} \frac{1}{D^2} + \frac{Z}{D}
	$$
	Thus
	$$
	\mathbb{E}(X^2 \mid Z) = \frac{Z^2 - Z}{2 N^2 D^2} + \frac{Z}{N^2 D}
	$$
	Now we can use the fact that $\mathbb{E}(Z) = N$ and $\mathbb{E}(Z^2) = 
	2N^2 - N$. 
	%(This second moment can be computed easily with a recurrence relation).
	$$
	\mathbb{E}(X^2) = \frac{(2N^2 - N) + N}{2N^2 D^2} + \frac{N}{N^2 D}
	= \frac{1}{D^2} + \frac{1}{ND}
	$$
	Finally, this means that the second moment of the revenue is 
	$$
	\mathbb{E}(\strategicR^2) = \frac{\alpha}{D^2} + \frac{\alpha}{ND}
	$$
	which in turn tells us the variance is
	$$
	var(\strategicR) = \mathbb{E}(\strategicR^2) - \mathbb{E}(\strategicR)^2 = 
	\frac{1}{D^2}(\alpha - \alpha^2) + \frac{\alpha}{ND}
	$$
\end{proof}

In deterministic PPLNS, block reward variance can be computed in an identical 
fashion, and it is $\frac{\alpha}{2D^2} + \frac{\alpha}{ND} - \frac{\alpha^2}{D^2} - 
\frac{\alpha}{2ND^2}$. Typically, pools have $N = 2D$, in which case the PPLNS 
variance becomes $ \frac{1}{D^2} (\alpha - \alpha^2) - \frac{\alpha}{4D^3}$. For 
this difficulty setting, RPPLNS block reward variance becomes 
$\frac{1}{D^2}(\alpha - \alpha^2) + \frac{\alpha}{2D^2}$. Though this is more than 
with standard PPLNS, this still vanishes at the same asymptotic rate of 
$O(1/N^2)$ when $N = \Theta(D)$.

\subsection{State Space Reduction}
\label{sec:state-space-reduction}

In this section, we prove that RPPLNS generally results in an exponential 
reduction in the state space $\mathcal{M}$ needs to implement the protocol.

\begin{restatable}{theorem}{statespace}
	\label{thm:state-reduction}
	Suppose that $\mathcal{M}$ is a mining pool with $m$ miners. In order to 
	implement PPLNS with parameter $N$, $\mathcal{M}$ needs at least $m^N$ 
	states, whereas $\mathcal{M}$ can implement RPPLNS with parameter $N$ 
	using at most $N\frac{(N+m-2)^{m-1}}{(m-1)!}$ states.  
\end{restatable}

\begin{proof}
	
	The first part of the theorem is clear from the definition of PPLNS. We focus 
	on the state that PPLNS needs once the queue is full after the initial $N$ 
	turns. Essentially $\mathcal{M}$ needs to keep track of a sliding window of 
	the ownership of the shares sent to it. Since the messages sent to 
	$\mathcal{M}$ can only come from $[m]$, there are thus $m^N$ many such 
	sequences.    
	
	As for the second part of the claim, we begin by focusing on the states used 
	once the bag is full after $N$ turns. We notice that the set of all possible 
	partitions of $N$ bag shares into $m$ owners is in one-to-one 
	correspondence with non-negative integral points of the $m-1$ simplex 
	scaled by a factor of $N$: $N \Delta^{m-1} = \{ x \in \mathbb{R}^{m-1} \ | \ 
	\forall i \  x_i \geq 0, \text{ and } \|x\|_1 \leq N\}$. In \cite{yau2006upper}, this is 
	shown to be upper bounded by $\frac{(N+m-2)^{m-1}}{(m-1)!}$. For the first 
	$N$ transitions of the protocol, we notice that the set of all possible partitions 
	of $k < N$ bag shares into $m$ owners is in one-to-one correspondence with 
	non-negative integral points of $k\Delta^{m-1} \subsetneq N\Delta^{m-1}$, 
	hence by a union bound we get $|\bigcup_{k=1}^N k\Delta^{m-1}| \leq 
	N|N\Delta^{m-1}| \leq N\frac{(N+m-2)^{m-1}}{(m-1)!}$ as desired.
	
\end{proof}

To put Theorem \ref{thm:state-reduction} into perspective, it is often the case 
that $m < < N$, so if we let $m = O(1)$ and consider the cardinality of the state 
space as a function of $N$, we get that for PPLNS this is exponential in $N$: 
$O(m^N)$ and for RPPLNS this is only polynomial in $N$: $O(N^m)$.  Similarly, 
notice in addition to this, how in PPLNS, if $\mathcal{M}$ wishes to 
communicate to a specific miner the state of his shares, this requires $N$ bits, 
as the miner needs to know the position of all his shares in the queue. On the 
other hand, RPPLNS only needs to communicate $\log N$ bits, due to the fact 
that it suffices to tell miners how many shares they have in the bag, as there is 
no sequentiality in the bag.

\subsection{Robustness to Pool-hopping}
\label{sec:pool-hop-proof}

In this section we show that if a miner is given the choice between mining with 
two different RPPLNS pools, then in expectation he will always earn the same 
block reward, irrespective of the initial state of his shares in each pool and how 
he may choose to partition his mining between said pools. In order to prove 
this, suppose that there is a single miner $\strategic$ of hash power $\alpha$, 
and two RPPLNS pools, $\mathcal{M}_i$ for $i \in \{1,2\}$. Furthermore, 
suppose that each $\mathcal{M}_i$ has bag size $N_i$, hash power $\beta_i$ 
and difficulty $D_i$, so that on average, $\mathcal{M}_i$ mines one block every 
$D_i$ shares.

Previously we studied fairness of a single mining pool, and hence we could 
model share/block mining as a turn-based process, where each turn a share is 
found, and ownership is dictated by agent hash rates. Having a turn-based 
approach with multiple pools is more fickle however, since pools may have 
different share difficulties, and hence the expected duration of such a turn in the 
continuous-time mining process may be different. For this reason, in this section 
we instead consider a continuous-time mining process. 

As is standard, we assume that share/block mining follows a Poisson process. 
We further assume that time units are normalised so that the expected time it 
takes for a block to be mined by the entire mining ecosystem is one time unit. 
Given this assumption, it follows that if a $\eta \in [0,1]$ proportion of the global 
mining hash power is dedicated to $\mathcal{M}_i$ for $T$ time units, then in 
expectation $\eta D_i T$ shares are found in $\mathcal{M}_i$ for $i \in \{1,2\}$. 
Given these observations, we are in a position to prove that RPPLNS is 
pool-hop-proof.

\begin{theorem}
	\label{thm:RPPLNS-pool-hop}
	Suppose that $\strategic$ has $A_i$ shares in $\mathcal{M}_i$ for $i \in 
	\{1,2\}$ at time $t = 0$. Furthermore, suppose that $\{I_1,...,I_k\}$ is an 
	arbitrary finite disjoint collection of closed intervals of $[0,T]$, such that $I_\ell 
	= [x_\ell,y_\ell]$, where $T$ is arbitrary as well. Let $T_2 = \cup_{i=1}^k I_i$ 
	and $T_1 = [0,T] \setminus T_2$ and suppose that $\strategic$ mines for 
	$\mathcal{M}_i$ for each $T_i$, $i \in \{1,2\}$, respectively. Then the expected 
	lifetime block reward of $\strategic$ for residual shares $A_1,A_2$ and 
	shares/blocks mined in $[0,T]$ is $\frac{A_1}{D_1} + \frac{A_2}{D_2} + \alpha 
	T$.   
\end{theorem}

\begin{proof}
	
	We first derive the expected payoff $\strategic$ obtains from $A_i$ residual 
	shares in $\mathcal{M}_i$. Since the pool manager will remove one share out 
	of queue randomly whenever a new share arrives, the expected number of 
	these $A_i$ residual shares will shrink exponentially as new shares come in. 
	After $n$ new shares, the expected number of residual shares will be $A_i 
	\left(\frac{N_i-1}{N_i}\right)^{n-1}$. Furthermore, every time a share is found, it 
	has a $\frac{1}{D_i}$ chance of being a full solution, in which case every share 
	$\strategic$ has in the bag pays $\frac{1}{N_i}$ block reward. This means that 
	the total expected lifetime payoff from these $A_i$ shares can be written as,
	\begin{eqnarray*}
		&&\frac{1}{D_i} \frac{A_i}{N_i} \left(1 + \frac{N_i-1}{N_i} + 
		\left(\frac{N_i-1}{N_i}\right)^{2} + \dots \right) \\
		&&=\frac{1}{D_i} \frac{A_i}{N_i} \frac{1}{1-\frac{N_i-1}{N_i}} = \frac{A_i}{D_i}.
	\end{eqnarray*}
	
	It only remains to compute the expected lifetime block reward $\strategic$ 
	obtains from shares/blocks mined during $[0,T]$. In what follows we let $|T_2| 
	= \sum_{i=1}^k|y_i - x_i|$ and $|T_1| = T - |T_2|$ denote the amount of time 
	$\strategic$ mines in $\mathcal{M}_i$ during $[0,T]$. We note that during 
	$T_1$, a total of $\alpha + \beta_1$ hash power is being contributed to 
	$\mathcal{M}_1$. This means that $\mathcal{M}_1$ mines $D_1 |T_1| (\alpha + 
	\beta_1)$ shares in expectation during during this time. Furthermore, every 
	share has a $\frac{\alpha}{\alpha + \beta}$ probability of belonging to 
	$\strategic$, therefore $\strategic$ finds $D_1 \alpha |T_1|$ shares for 
	$\mathcal{M}_1$ in $T_1$. Furthermore, as we saw in Theorem 
	\ref{thm:RPPLNS-fair-variance}, each of these shares earns $\frac{1}{D_1}$ 
	block reward over its lifetime in expectation, therefore the total lifetime block 
	rewards from shares found for $\mathcal{M}_1$ in $[0,T]$ is $\alpha |T_1|$ in 
	expectation. In an identical fashion we can show that the lifetime rewards for 
	shares found for $\mathcal{M}_2$ in $[0,T]$ is $\alpha |T_2|$. Putting 
	everything together, this means that the total expected lifetime reward for the 
	residual shares $A_1$, $A_2$ and newly found shares in $[0,T]$ is 
	$\frac{A_1}{D} + \frac{A_2}{D} + \alpha T$ as desired.
	
\end{proof}

The fact that the expected lifetime reward is an expression that is independent of 
the partition, $T_1$, precisely implies that RPPLNS is robust to pool-hopping. 
Furthermore, it is straightforward to generalise Theorem 
\ref{thm:RPPLNS-pool-hop} to encompass a choice of hopping between an 
arbitrary number of pools, yielding the same robustness to pool-hopping in this 
setting.

\subsection{Steady State of Honest Pool for RPPLNS}
\label{sec:steady-state}

In this section we show that the number of shares $\strategic$ has in the bag of 
size $N$ can be modelled as an ergodic Markov chain. We explicitly derive the 
steady state distribution and use it to compute the expected number of shares 
the miner has in the bag in the honest steady state. 

\subsubsection{The Markov Chain}

Let us suppose that miner $\strategic$ currently has $i \in \{0,...,N\}$ shares in 
the bag and is honest. With probability $\alpha$, $\strategic$ finds a block/share 
and publishes it to the pool operator. Conditional to this, with probability 
$\frac{i}{N}$, the number of shares stays the same, $i \rightarrow i$, and with 
probability $\frac{N-i}{N}$, the number of shares increases by one, $i \rightarrow 
i+1$. On the other hand, $\honest$ finds a block/share and publishes it with 
probability $\beta$. Conditional to this, with probability $\frac{i}{N}$, the number 
of shares goes down by one, $i \rightarrow i-1$, and with probability 
$\frac{N-i}{N}$, the number of shares stays the same. Clearly, this induces an 
ergodic Markov chain on the state space $S = \{0,...,N\}$. Suppose that $\pi = 
(\pi_i)_{i=1}^N$, is the unique honest steady state distribution over $S$. 

\begin{restatable}{theorem}{steadystate}
	For the steady state distribution $\pi$ of miner shares in the RPPLNS bag: 
	$$
	\pi_k = \frac{\binom{N}{k}\left(\frac{\alpha}{\beta} \right)^k}{\left( 1 + 
	\frac{\alpha}{\beta} \right)^N}
	\text{, for }k = 0,...,N.
	$$
	Furthermore, the expected number of shares in the bag under $\pi$ is 
	$$ 
	N \left(\frac{\alpha}{\alpha + \beta} \right).
	$$
	
\end{restatable}

\begin{proof}
	From the topology of the chain, we can make a cut-set between $S_i = 
	\{0,...,i\}$ and $S_i^c = \{i+1,...,N\}$ for $i = 0,...,N-1$. From steady state 
	conditions it follows that $\pi$ must satisfy $\pi_i P_{i\rightarrow i+1} = 
	\pi_{i+1} P_{i+1 \rightarrow i} $ for all such cut-sets. It follows that 
	$P_{i\rightarrow i+1} = \alpha \frac{N-i}{N}$, and $P_{i+1 \rightarrow i} = \beta 
	\frac{i+1}{N}$. As a consequence, we obtain $\frac{\pi_{i+1}}{\pi_i} = 
	\frac{\alpha}{\beta} \left( \frac{N-i}{i+1} \right) = f(i)$. Using this expression we 
	know that $\pi_{i+1} = f(i) \pi_i$ for all $i = 0,...,N-1$, consequently $\pi_k = 
	\left(\prod_{j=0}^{k-1}f(j)\right) \pi_0 = \pi_0 \binom{N}{k} ( \frac{\alpha}{\beta} 
	)^k$ for all $k$. The final step in obtaining the steady state is normalising the 
	sum of all terms, which corresponds to $\sum_{k=0}^N \pi_k = 
	\pi_0\sum_{k=0}^N \binom{N}{k} (\frac{\alpha}{\beta})^k$. Using the 
	well-known expansion for $(1+x)^n$, we get that this is $\pi_0 \left( 1 + 
	\frac{\alpha}{\beta} \right)^N$. As a consequence, it follows that for $\pi$ to 
	be a steady state, we need $\pi_0 = \left( 1 + \frac{\alpha}{\beta} \right)^{-N}$ 
	and consequently:
	$$
	\pi_k = \frac{\binom{N}{k}\left(\frac{\alpha}{\beta} \right)^k}{\left( 1 + 
	\frac{\alpha}{\beta} \right)^N}
	\text{, }k = 0,...,N
	$$
	
	Now  we wish to compute the expected number of shares in the bag in the 
	steady state by using our expression $\pi$. To do so, let us first define $f(x) = 
	(1 + x)^n$ for $n \in \mathbb{N}$ and $x \in \mathbb{R}$. We know that $f(x) = 
	\sum_{k=0}^n \binom{n}{k}x^k$. Therefore, $nx(1+x)^{n-1} = x \frac{d}{dx} f(x) 
	= \sum_{k=0}^n \binom{n}{k} k x^k$. We use this formula to compute the 
	expectation over $\pi$, $\mathbb{E}_{\pi} = \sum_{k=0}^N k \pi_k$. 
	
	\begin{equation*}
		\begin{split}
			\mathbb{E}_{\pi} = \sum_{k=0}^N k \pi_k &= \left(1 + \frac{\alpha}{\beta} 
			\right)^{-N} \sum_{k=0}^N  \binom{N}{k} k \left( \frac{\alpha}{\beta} 
			\right)^k \\
			& = \left( 1 + \frac{\alpha}{\beta} \right)^{-N} N \left( \frac{\alpha}{\beta} 
			\right) \left( 1 + \frac{\alpha}{\beta} \right)^{N-1} \\
			& = N \left(\frac{\alpha}{\alpha + \beta} \right)
		\end{split}
	\end{equation*}
	
\end{proof}

As a first consequence, note that this also constitutes a proof that RPPLNS is 
fair if everyone is honest. The reason for this is that in each turn, we suppose 
that one of either $\strategic,\honest$ or $\nonpool$ find a share with probability 
$\alpha, \beta, \gamma$ respectively, and each share has a further probability of 
$\frac{1}{D}$ of being a full solution. Thus the probability that a single turn ends 
up paying agents from the pool is precisely $\frac{\alpha + \beta}{D}$. As a 
consequence, the expected payment to $\strategic$ in the steady state is $\left( 
\mathbb{E}_\pi \right) \left( \frac{\alpha + \beta}{ND} \right) = \frac{\alpha}{D}$, 
which is precisely the expected per-turn payment of $\strategic$ mining solo.

\subsection{When is Honest Mining a Dominant Strategy}
\label{sec:recursion}

We recall that we are in the setting of a single pool miner being strategic. In 
other words we have $\strategic$, $\honest$ and $\nonpool$ of hash powers 
$\alpha$, $\beta$ and $\gamma$ respectively. We wish to find conditions such 
that $\strategic$ is honest (i.e. publishes shares and blocks to $\mathcal{M}_R$ 
immediately). Given this setting, at any given moment, we describe the state of 
$\strategic$ within the system with a tuple $(\ell, s, b) \in [N]^2 \times \{0,1\}$. 
$\ell$ represents the shares $\strategic$ has in the bag of the pool, $s$ 
represents the private shares of $\strategic$, and $b$ represents whether 
$\strategic$ has a private block or not. Note that $b \le 1$ is without loss of 
generality, as the extra blocks mined are invalidated once one of them is 
released. Requiring $s \le N$ is not without loss of generality, since technically a 
selfish miner could hoard an infinite amount of shares. However, it is reasonable 
to assume that $\alpha < 0.5$ (to prevent scenarios where double spend attacks 
are possible in the first place). Hence it is exceedingly unlikely that $s$ will reach 
values higher than $N$, given a typical difficulty such as $D = N/2$. 
Furthermore, as we will soon see, any time an agent other then $\strategic$ finds 
a block, private shares are invalidated, hence further making it difficult to hoard a 
large amount of private shares. The benefit of upper bounding $s$ and $b$ is 
significant, as the state space becomes finite.

\subsubsection{Recurrence Relation}

Let $g_k(\ell,s,b)$ be the maximum revenue a strategic miner can obtain when 
acting optimally starting at $g(\ell,s,b)$, after exactly $k$ shares in total have 
been mined by his pool. The parameter $k$ is necessary for the recursion to 
prevent it from being open ended.

\begin{equation}\label{eq:recurrence}
	g_k(\ell,s,b) = max 
	\begin{cases}
		\text{A} &\quad (\strategic\text{ waits and } s < N) \\
		\text{B} &\quad (\strategic\text{ waits and } s = N) \\
		\text{C} &\quad (\strategic\text{ publishes a share)}\\
		\text{D} &\quad (\strategic\text{ publishes a block)}\\
	\end{cases}
\end{equation}
We present the expressions for $A, B, C$ and $D$ in the following sections.

\subsubsection{$\strategic$ Waits and $s < N$}

When $\strategic$ waits and $s < N$, there are 6 immediate outcomes:
\begin{itemize}
	\item $\strategic$ gets a block. Call this event ``$a$'' which happens with 
	probability $\alpha \left(\frac{1}{D} \right)$. The resulting state is $(\ell,s,1)$.
	\item $\strategic$ gets a share. Call this event ``$b$'' which happens with 
	probability $\alpha \left(\frac{D-1}{D} \right)$. The resulting state is 
	$(\ell,s+1,b)$. 
	\item $\honest$ gets a block. Call this event ``$c$'' which happens with 
	probability $\beta \left(\frac{1}{D} \right)$. Since $\honest$ is honest, he 
	publishes this block and the private shares and blocks of $\strategic$ are 
	rendered obsolete. This means that with probability $\frac{\ell}{N}$ a share of 
	$\strategic$ is kicked out of the bag resulting in state $(\ell-1,0,0)$ and 
	$\strategic$ getting paid $\frac{\ell-1}{N}$. With probability $\frac{N-\ell}{N}$ a 
	share of $\honest$ is kicked out of the bag resulting in state $(\ell,0,0)$ and 
	$\strategic$ getting paid $\frac{\ell}{N}$.
	\item $\honest$ gets a share. Call this event ``$d$'' which happens with 
	probability $\beta \left(\frac{D-1}{D} \right)$. Since $\honest$ is honest, he 
	publishes this share. This means that with probability $\frac{\ell}{N}$ a share 
	of $\strategic$ is kicked out of the bag and the resulting state is $(\ell-1,s,b)$, 
	and with probability $\frac{N-\ell}{N}$ a share of $\honest$ is kicked out of the 
	bag and the resulting state is $(\ell,s,b)$.
	\item $\nonpool$ gets a block. Call this event ``$e$'' which happens with 
	probability $\gamma \left(\frac{1}{D} \right)$. Since $\nonpool$ is honest, he 
	publishes this block and the private shares and blocks of $\strategic$ are 
	rendered obsolete. This means that the resulting state is $(\ell,0,0)$.
	\item $\nonpool$ gets a share. Call this event ``$f$'' which happens with 
	probability $\gamma \left(\frac{D-1}{D} \right)$. In this case the state remains at 
	$(\ell,s,b)$
\end{itemize}

We can now get an expression for $A$ in the recursion for $g_k$:
\begin{align*}
	A = \quad&
	\alpha \left(\frac{1}{D} \right) g_{k-1}(\ell,s,1) 
	+ \alpha \left(\frac{D-1}{D} \right) g_{k-1}(\ell,s+1,b)\\ 
	&\quad+ \beta \left(\frac{1}{D} \right) \left( \frac{\ell}{N} g_{k-1}(\ell-1,0,0) + 
	\frac{N-\ell}{N} g_{k-1}(\ell,0,0) + \frac{\ell}{N} \right)\\
	&\quad+\beta \left(\frac{D-1}{D} \right) \left( \frac{\ell}{N} g_{k-1}(\ell-1,s,b) + 
	\frac{N-\ell}{N} g_{k-1}(\ell,s,b) \right)\\ 
	&\quad+ \gamma \left(\frac{1}{D} \right) g_{k-1}(\ell,0,0) 
	+ \gamma \left(\frac{D-1}{D} \right) g_{k-1}(\ell,s,b)
\end{align*}

\subsubsection{$\strategic$ Waits and $s = N$}
This case is almost identical to the above. The only difference is in event $b$: 
when $\strategic$ gets a share. Since we don't want to consider an infinite state 
space, we will assume that upon obtaining $N$ private shares, $\strategic$ will 
always publish subsequent mined shares. Since $N$ is usually quite large, the 
probability of even obtaining this state is small. 

\begin{itemize}
	\item $\strategic$ gets a share. Call this event ``$b$'' which happens with 
	probability $\alpha \left(\frac{D-1}{D} \right)$. $\strategic$ will publish this 
	share. Hence the resulting state is $(\ell,s,b) = (\ell,N,b)$.  
\end{itemize}

With identical reasoning to the above, we get the following expression for $B$:
\begin{align*}
	A = \quad&
	\alpha \left(\frac{1}{D} \right) g_{k-1}(\ell,N,1) 
	+ \alpha \left(\frac{D-1}{D} \right) g_{k-1}(\ell,N,b)\\ 
	&\quad+ \beta \left(\frac{1}{D} \right) \left( \frac{\ell}{N} g_{k-1}(\ell-1,0,0) + 
	\frac{N-\ell}{N} g_{k-1}(\ell,0,0) + \frac{\ell}{N} \right)\\
	&\quad+\beta \left(\frac{D-1}{D} \right) \left( \frac{\ell}{N} g_{k-1}(\ell-1,N,b) + 
	\frac{N-\ell}{N} g_{k-1}(\ell,N,b) \right)\\ 
	&\quad+ \gamma \left(\frac{1}{D} \right) g_{k-1}(\ell,0,0) + \gamma 
	\left(\frac{D-1}{D} \right) g_{k-1}(\ell,N,b)
\end{align*}

\subsubsection{$\strategic$ Publishes a Share}

If the current state of the system is $(\ell,s,b)$, then by publishing a share, there 
is a $\frac{\ell}{N}$ probability that a share belonging to $\strategic$ is kicked 
out of the bag, resulting in the state $(\ell,s-1,b)$. There is a $\frac{N-\ell}{N}$ 
probability that a share belonging to $\honest$ is kicked out of the bag, 
resulting in the state $(\ell+1,s-1,b)$. Note that a share is not actually mined, 
therefore the recurrence does not involve $g_{k-1}$. Putting this together, we get 
the following value of $C$:

$$
C = \frac{\ell}{N} g_{k}(\ell,s-1,b) + \frac{N-\ell}{N} g_{k}(\ell+1,s,b)
$$

\subsubsection{$\strategic$ Publishes a Block}

If the current state of the system is $(\ell,s,b)$, then by publishing a block, there 
is a $\frac{\ell}{N}$ probability that a share belonging to $\strategic$ is kicked 
out of the bag, resulting in the state $(\ell,s-1,b)$ and $\strategic$ earning 
$\frac{\ell}{N}$. There is a $\frac{N-\ell}{N}$ probability that a share belonging to 
$\honest$ is kicked out of the bag, resulting in the state $(\ell+1,s-1,b)$ and 
$\strategic$ earning $\frac{\min (N,\ell+1 )}{N}$. Putting this together, we get the 
following value of $D$:

$$
D= \frac{\ell}{N}\left(\frac{\ell}{N} + g_{k}(\ell,0,0)\right) + 
\frac{N-\ell}{N}\left(\frac{\min(N,\ell+1)}{N}+ g_{k}(\ell+1,0,0)\right)
$$

\subsubsection{Checking for Equilibrium Conditions}
Observe that since the state space is finite the value obtained per `step' $k$ is at 
most $1$, the following limit is well defined for all $(l,s,b)$:
$$
\phi(l,s,b) = \lim_{k\rightarrow \infty} \frac{g_k(l,s,b)}{k}.
$$
This $\phi(l,s,b)$ can be thought of as the potential of state $(l,s,b)$. Comparing 
$\phi$ values of different states, we can deduce where the miner obtains higher 
payoffs.To tell for which $\alpha, \beta, N$, and $D$ honest mining is a 
dominant strategy, we need to check the following conditions:
\begin{equation}\label{ineq:dsic_share}
	\phi(\ell+1,0,0) \ge \phi(\ell,1,0) \text{ when } (\ell) \in [N].
\end{equation}
\begin{equation}\label{ineq:dsic_block}
	\frac{N-\ell}{N} \cdot \phi(\ell+1,0,0) +
	\frac{\ell}{N}\cdot\phi(\ell,0,0) \ge \phi(\ell,0,1) \text{ when } (\ell) \in [N],
\end{equation}
where \cref{ineq:dsic_share} and \cref{ineq:dsic_block} show that releasing 
shares and blocks respectively is at least as good as hoarding them.

\subsubsection{Computational Setup}
Unfortunately we cannot solve exactly for $\phi$, but we can approximate it well 
enough given high enough $k$ to observe if strategic mining is an option. We 
will use $g_k$ as defined in \cref{eq:recurrence} with one difference: if state 
$(l,N,b)$ is reached the miner is rewarded $1$ and jumps to state $(N,0,0)$. This 
incentivizes hoarding blocks, since hoarding enough of them results in an instant 
maximum payoff. The reason for this is to strengthen the experimental results: if 
honest mining is viable in this setting, it is more likely that it is actually the best 
option.

Having computed $g_k$ for large enough $k$, we can then check 
\Cref{ineq:dsic_share} as it is. For \Cref{ineq:dsic_block} we need the following 
small adjustment:
\begin{align*}
	&\frac{N-\ell}{N}\left(\frac{\ell+1}{N} + g_k(\ell+1,0,0)\right)+
	\frac{\ell}{N}\left(\frac{\ell}{N} + g_k(\ell,0,0)\right) \\ 
	&> g_k(\ell,0,1) \text{ when } (\ell) \in [N].
\end{align*}

This is because as $k$ goes to infinity, the immediate rewards $(\ell+1)/N$ and 
$\ell/N$ vanish compared to the potential of the three states involved (as $\phi$ 
is defined by diving with $k$). Since we are working with finite $k$, we need to 
plug them back in.

\subsubsection{Results}

To understand the behaviour of the miners, we explicitly compute the truncated 
recursion in Section \ref{sec:recursion} for $N=1000, D=500$ and $k=150$. Each 
of the following graphs shows the best action for $\strategic$ given an 
\emph{initial} fraction of shares in the bag, which we call $F$. For $F \le 0.35$, 
we have:
\begin{figure*}
	\centering
	\scalebox{0.33}{\input{f=0.05.pgf}}
	\scalebox{0.33}{\input{f=0.10.pgf}}
	\scalebox{0.33}{\input{f=0.15.pgf}}
	\scalebox{0.33}{\input{f=0.20.pgf}}
	\scalebox{0.33}{\input{f=0.25.pgf}}
	\scalebox{0.33}{\input{f=0.30.pgf}}
\end{figure*}

In this regime of $F$ values, the strategic miner begins by controlling but a few 
shares of the bag, and his best option is often to \emph{hoard} a block. This is 
because they expect to be able to add a couple more shares to the bag before 
someone else manages to mine another block and invalidate their private block. 
In Section \ref{sec:strategic-hoarding} we formally demonstrate that for common 
difficulty settings of $D = \Theta(N)$ (typically $D = N/2$), if $\alpha = 
\Omega(1/N)$, and a pool miner has no shares in the bag, then hoarding a block 
for a single turn dominates behaving honestly. This behaviour is not unique to 
RPPLNS though, as we also demonstrate in the same section that for similar 
hash rates, if a pool miner is facing an empty queue in PPLNS, they will hoard a 
private block with hopes of finding a share in subsequent turns. 

What is most interesting though, is that if we recall our derivations of the steady 
state of bag shares from Section \ref{sec:steady-state}, strategic block hoarding 
only occurs at hash rates and $F$ values such that $F$ is in fact much less than 
the expected number of bag shares in the steady state. This suggests that if all 
miners behave honestly, a single pool miner is more likely to find himself at a 
state where mining honestly is a dominant strategy. 

On the other hand, for initial share distributions, $0.35 \le F \le 0.70$, our 
recursions suggest that honest mining is the best option throughout. We have 
omitted graphs of these cases since they simply paint the simplex red entirely. 
Finally, when $F \geq 0.70$, as in the low initial share regime, we see that 
strategic behaviour arises once more, though this time in the form of hoarding 
shares rather than blocks. For example, it is not difficult to see that if $F = 1$, 
then in both PPLNS and RPPLNS hoarding a share dominates publishing it 
immediately, as doing so has no effect on the state of the pool.

Once again, share hoarding becomes a better strategy at initial share 
distributions that are far from the expected share distribution of the steady state 
of bag shares from Section \ref{sec:steady-state}. This once again suggests that 
states where miners act strategically are more rare than those where a miner is 
honest.

\begin{figure*}
	\centering
	\scalebox{0.33}{\input{f=0.75.pgf}}
	\scalebox{0.33}{\input{f=0.80.pgf}}
	\scalebox{0.33}{\input{f=0.85.pgf}}
	\scalebox{0.33}{\input{f=0.90.pgf}}
	\scalebox{0.33}{\input{f=0.95.pgf}}
	\scalebox{0.33}{\input{f=1.00.pgf}}
\end{figure*}

\section{Strategic Block Hoarding}
\label{sec:strategic-hoarding}

As seen in our empirical results, there are parameter settings in which we see 
strategic pool mining in RPPLNS. In this section we theoretically justify why this 
is the case for block hoarding in particular. 

\subsection{PPLNS}
We assume that there is a single strategic pool miner, $\strategic$ with hash 
power $\alpha$. This miner operates within a PPLNS mining pool with hash 
power $\alpha + \beta$, where the honest miners in the pool can be considered 
as a single honest pool miner $\honest$. Finally, we let $\nonpool$ represent all 
other honest miners in a system with collective hash power $\gamma = 1 - 
\alpha - \beta$.

Let us suppose that $\strategic$ has no shares in the pool queue and that he 
currently holds a valid block he could send to the pool operator. We compare 
two mining strategies over a two-element decision window. If $\strategic$ is 
honest, he will immediately publish this block, wait one more turn to see who 
receives a share/block, and subsequently act honestly (publish a new 
share/block immediately if found). On the other hand, we consider a one-time 
strategic deviation by $\strategic$ as follows: $\strategic$ hoards the current 
block he has, and waits one more turn to see what occurs. If he receives a 
share, he publishes the share before his held block. In all other scenarios he acts 
honestly. We call the honest strategy $H$ and the strategic deviation $S$. 

\begin{theorem}
	\label{thm:PPLNS-strategic}
	If $\alpha > \frac{N+D-1}{(D-1)^2}$, then $S$ is a strictly dominates $H$. 
\end{theorem}

\begin{proof}
	How $H$ and $S$ perform depends on precisely one of 6 cases:
	\begin{itemize}
		\item $\strategic$ finds a block in the next turn. This case occurs with 
		probability $\alpha \frac{1}{D}$.
		\item $\strategic$ finds a share in the next turn. This case occurs with 
		probability $\alpha \frac{D-1}{D}$.
		\item $\honest$ finds a block in the next turn. This case occurs with 
		probability $\beta \frac{1}{D}$.
		\item $\honest$ finds a share in the next turn. This case occurs with 
		probability $\beta \frac{D-1}{D}$.
		\item $\nonpool$ finds a block in the next turn. This case occurs with 
		probability $\gamma \frac{1}{D}$.
		\item $\nonpool$ finds a share in the next turn. This case occurs with 
		probability $\gamma \frac{D-1}{D}$.
	\end{itemize}
	First we focus on the expected payoffs $H$ receives in each of these cases:
	\begin{itemize}
		\item Here $\strategic$ publishes two blocks back to back. This results in 
		an immediate reward of $\frac{3}{N}$. The older block is still eligible for 
		payment for $N-2$ turns and the latter for $N-1$ turns, hence their 
		expected future payment is $\frac{N-2}{ND} + \frac{N-1}{ND}$. Overall we 
		denote this expected revenue by $H_1 = \frac{3}{N} + \frac{N-2}{ND} + 
		\frac{N-1}{ND}$.
		\item Here $\strategic$ publishes a block and then a share. This results in 
		an immediate reward of $\frac{1}{N}$. As before, the older share is alive for 
		$N-2$ more turns and the latter for $N-1$, hence their expected future 
		payment is $\frac{N-2}{ND} + \frac{N-1}{ND}$. Overall we denote this 
		expected revenue by $H_2 = \frac{1}{N} + \frac{N-2}{ND} + \frac{N-1}{ND}$.
		\item Here $\strategic$ publishes a block and $\honest$ subsequently 
		publishes a block. This results in an immediate reward of $\frac{2}{N}$. The 
		first block is alive for another $N-2$ turns, hence its expected future 
		payment is $\frac{N-2}{ND}$. Overall we denote this expected revenue by 
		$H_3 = \frac{2}{N} + \frac{N-2}{ND}$.
		\item Here $\strategic$ publishes a block and $\honest$ subsequently 
		publishes a share. This results in an immediate reward of $\frac{1}{N}$. The 
		first block is alive for another $N-2$ turns, hence its expected future 
		payment is $\frac{N-2}{ND}$.  Overall we denote this expected revenue by 
		$H_4 = \frac{1}{N} + \frac{N-2}{ND}$.
		\item Here $\strategic$ publishes a block and $\nonpool$ subsequently 
		finds a block. The latter event has no bearings on the pool, hence this is 
		equivalent to $\strategic$ publishing a single block, which gives an 
		immediate reward of $\frac{1}{N}$ and since this block is alive for another 
		$N-1$ turns, an expected future reward of $\frac{N-1}{ND}$. Overall we 
		denote this expected revenue by $H_5 = \frac{1}{N} + \frac{N-1}{ND}$.
		\item Here $\strategic$ publishes a block and $\nonpool$ subsequently 
		finds a share. The latter event has no bearings on the pool, hence this is 
		equivalent to $\strategic$ publishing a single block, which gives an 
		immediate reward of $\frac{1}{N}$ and since this block is alive for another 
		$N-1$ turns, an expected future reward of $\frac{N-1}{ND}$. Overall we 
		denote this expected revenue by $H_6 = \frac{1}{N} + \frac{N-1}{ND}$.
		
	\end{itemize}
	
	Given our previous expressions, we can write the expected revenue 
	$\strategic$ obtains from using $H$ by 
	\begin{align*}
		R_H &= H_1 \left( \frac{\alpha}{D} \right) + H_2 \left( \frac{\alpha(D-1)}{D} 
		\right)
		+ H_3 \left( \frac{\beta}{D} \right)\\
		&\quad+ H_4 \left( \frac{\beta(D-1)}{D} \right)
		+ H_5 \left( \frac{\gamma}{D} \right) + H_6 \left( \frac{\gamma(D-1)}{D} \right)
	\end{align*}
	
	We proceed to compute revenues for the strategic deviation $S$:
	
	\begin{itemize}
		\item Here $\strategic$ hoards a block and then receives a new block. This 
		forcibly invalidates a single block hence, hence $\strategic$ only gets an 
		immediate reward of $\frac{1}{N}$ and subsequently this block lives 
		another $N-1$ turns to obtain an expected lifetime revenue of 
		$\frac{N-1}{ND}$. Overall we denote this expected revenue by $S_1 = 
		\frac{1}{N} + \frac{N-1}{ND}$.
		\item Here $\strategic$ hoards a block and then recieves a share, 
		publishing the share first and then the block. This results in an immediate 
		reward of $\frac{2}{N}$. Subsequently the latter share is alive for $N-2$ 
		turns and the block for $N-1$, hence the expected future payoff of these is 
		$\frac{N-2}{ND} + \frac{N-1}{ND}$. Overall we denote this expected revenue 
		by $S_2 = \frac{2}{N} + \frac{N-2}{ND} + \frac{N-1}{ND}$. 
		\item Here $\strategic$ hoards a block and $\honest$ publishes a block. 
		This invalidates the secret block and $\strategic$ gets no revenue. We 
		denote this by $S_3 = 0$.
		\item Here $\strategic$ hoards a block and $\honest$ publishes a share. 
		$\strategic$ only gets an immediate reward of $\frac{1}{N}$ and 
		subsequently this block lives another $N-1$ turns to obtain an expected 
		lifetime revenue of $\frac{N-1}{ND}$. Overall we denote this expected 
		revenue by $S_4 = \frac{1}{N} + \frac{N-1}{ND}$. 
		\item Here $\strategic$ hoards a block and $\nonpool$ publishes a block. 
		This invalidates the secret block and $\strategic$ gets no revenue. We 
		denote this by $S_5 = 0$.
		\item Here $\strategic$ hoards a block and $\honest$ publishes a share. 
		$\strategic$ only gets an immediate reward of $\frac{1}{N}$ and 
		subsequently this block lives another $N-1$ turns to obtain an expected 
		lifetime revenue of $\frac{N-1}{ND}$. Overall we denote this expected 
		revenue by $S_6 = \frac{1}{N} + \frac{N-1}{ND}$.
		
	\end{itemize}
	
	Given our previous expressions, we can write the expected revenue 
	$\strategic$ obtains from using $S$ by 
	\begin{align*}
		R_S = &S_1 \left( \frac{\alpha}{D} \right) + S_2 \left( \frac{\alpha(D-1)}{D} 
		\right) + S_3 \left( \frac{\beta}{D} \right)\\
		&+ S_4 \left( \frac{\beta(D-1)}{D} \right) + S_5 \left( \frac{\gamma}{D} \right) + 
		S_6 \left( \frac{\gamma(D-1)}{D} \right) 
	\end{align*}
	With these expressions in hand, we are interested in parameter values such 
	that $R_S > R_H$. 
	\begin{align*}
		R_S - R_H = &-\frac{\alpha}{D} \left(\frac{2}{N} + \frac{N-2}{DN} \right) + 
		\frac{\alpha(D-1)}{DN} - \frac{\beta}{D} \left( \frac{2}{N} + \frac{N-2}{ND} 
		\right)\\ 
		&+ \frac{\beta(D-1)}{D} \left( \frac{1}{ND} \right) - \frac{\gamma}{D} \left( 
		\frac{1}{N} + \frac{N-1}{ND} \right)
	\end{align*}
	We wish to show when this latter expression is strictly greater than 0. The 
	final term is $$- \frac{\gamma}{D} \left( \frac{1}{N} + \frac{N-1}{ND} \right) = - 
	\frac{\gamma}{D} \left( \frac{2}{N} + \frac{N-2}{ND} \right) + \frac{\gamma}{D} 
	\left( \frac{1}{N} - \frac{1}{ND} \right).$$ This means that we can join all terms 
	that have $\frac{2}{N} + \frac{N-2}{ND}$ and take advantage of the fact that 
	$\alpha + \beta + \gamma = 1$. This gives us the following expression:
	\begin{align*}
		R_S - R_H &= \frac{\alpha(D-1)}{ND} + \frac{\beta(D-1)}{ND^2} + 
		\frac{\gamma}{D} \left( \frac{1}{N} - \frac{1}{ND} \right)\\
		&\quad- \frac{1}{D} \left( \frac{2}{N} + \frac{N-2}{ND} \right) > 0
	\end{align*}
	We can now multiply both sides of the inequality by $\frac{ND}{D-1}$ to 
	isolate $\alpha$ and get
	$$
	\alpha + \frac{\beta}{D} + \frac{\gamma}{D-1} - \frac{\gamma}{D(D-1)} - 
	\frac{2}{D-1} - \frac{N-2}{D(D-1)} > 0
	$$
	Subsequent simplification gives
	\begin{align*}
		\alpha &> 
		\frac{2D + N - 2 + \gamma - \beta(D-1) - D\gamma}{D(D-1)} \\
		& = \frac{2D + N - 2 -(D-1)(\beta + \gamma)}{D(D-1)} \\
		& = \frac{2D + N - 2 -(D-1)(1 - \alpha)}{D(D-1)} \\
		& = \frac{N + D - 1}{D(D-1)} + \frac{\alpha}{D} \\
	\end{align*}
	Subtracting $\frac{\alpha}{D}$ from both sides we obtain
	\begin{align*}
		\alpha \left( \frac{D-1}{D} \right) &> \frac{N + D - 1}{D(D-1)} \\
		\alpha  &> \frac{N + D - 1}{(D-1)^2}\qedhere\\
	\end{align*}
\end{proof}
In most cases $N = 2D$, which gives a lower bound of $\alpha > \frac{3D - 
1}{(D-1)^2}$. $D$ is also usually quite large, hence this shows that in almost all 
cases miners will hoard blocks in the case they get lucky and find a block with 
an empty queue!

\subsection{RPPLNS}
Since RPPLNS is memoryless, we can extend the above analysis to the case 
where $\strategic$ miner has $k$ shares in the bag and holds a private block. 
We want to study the expected revenue from being honest for a single turn vs. 
waiting for the next turn in hopes of finding another share to publish before the 
withheld block.

\subsubsection{Important Terms}

In RPPLNS, if a specific share surives being pushed, it will be eligible for 
payment another $N-1$ pushes in expectation due to the memory-less property 
of their survival (it is a geometric random variable). If everyone is honest, each of 
these payment opportunities has a $\frac{1}{D}$ probability of actually paying a 
$\frac{1}{N}$ amount, hence shares that survive being pushed give $\strategic$ 
an expected revenue of $\frac{N-1}{ND}$. 

To simplify the expressions, we let $f_i^B(k,N)$ and $f_i^S(k,N)$ represent the 
expected utility $\strategic$ makes when miner $i \in \{1,2\}$ pushes a block or a 
share into the bag respectively:

\begin{equation*}
	\begin{split}
		f_1^B(k,N) &  = \frac{k}{N} \left( \frac{k}{N} + k \left(\frac{N-1}{ND} \right) 
		\right) + \frac{N-k}{N} \left( \frac{k+1}{N} + (k+1) \left(\frac{N-1}{ND} \right) 
		\right)  \\
		f_1^S(k,N) &  = \frac{k}{N} \left( k \left(\frac{N-1}{ND} \right) \right) + 
		\frac{N-k}{N} \left( (k+1) \left(\frac{N-1}{ND} \right) \right) \\
		f_2^B(k,N) &  = \frac{k}{N} \left( \frac{k-1}{N} + (k-1) \left(\frac{N-1}{ND} 
		\right) \right) + \frac{N-k}{N} \left( \frac{k}{N} + k \left(\frac{N-1}{ND} \right) 
		\right) \\
		f_2^S(k,N) &  = \frac{k}{N} \left( (k-1) \left(\frac{N-1}{ND} \right) \right) + 
		\frac{N-k}{N} \left( k \left(\frac{N-1}{ND} \right) \right) 
	\end{split}
\end{equation*}

With this in hand, we consider a two-turn scenario as with PPLNS. In one 
$\strategic$ has a block in hand, and is honest with this current block and with 
what happens the following term. This strategy is $H$ for honest. In the second 
case, $\strategic$ hoards this block in hand with hopes of mining a share the 
following turn and publishing the share before the block. After this minor 
strategic deviation for a turn though, $\strategic$ returns to being honest. This 
strategy is referred to by $S$. First we write the expected utility for each action 
for $H$
\begin{equation*}
	\begin{split}
		H_1 &  = \frac{k}{N} \left( \frac{k}{N} + f_1^B(k,N) \right) + \frac{N-k}{N} \left( 
		\frac{k+1}{N} + f_1^B(k+1,N) \right)  \\
		H_2 &  = \frac{k}{N} \left( \frac{k}{N} + f_1^S(k,N) \right) + \frac{N-k}{N} \left( 
		\frac{k+1}{N} + f_1^S(k+1,N) \right)  \\
		H_3 &  =  \frac{k}{N} \left( \frac{k}{N} + f_1^B(k,N) \right) + \frac{N-k}{N} \left( 
		\frac{k+1}{N} + f_2^B(k+1,N) \right) \\
		H_4 &  = \frac{k}{N} \left( \frac{k}{N} + f_2^S(k,N) \right) + \frac{N-k}{N} \left( 
		\frac{k+1}{N} + f_2^S(k+1,N) \right)  \\
		H_5 &  = f_1^B(k,N)  \\
		H_6 &  = f_1^B(k,N)  \\
	\end{split}
\end{equation*}

Next we write the expected utility for each action for $S$
\begin{align*} 
	S_1 &  = f_1^B(k,N)\notag  \\
	S_2 &  = \frac{k}{N} \left( f_1^B(k,N) \right) + \frac{N-k}{N} \left( f_1^B(k+1,N) 
	\right)\notag  \\
	S_3 &  =  0 \\
	S_4 &  = \frac{k}{N} \left( f_1^B(k-1,N) \right) + \frac{N-k}{N} \left( f_1^B(k,N) 
	\right)  \notag\\
	S_5 &  = 0 \notag \\
	S_6 &  = f_1^B(k,N) \notag 
\end{align*}

We remember the probability of each state
\begin{equation*} 
	\begin{split}
		p_1 = \frac{\alpha}{D}  &\qquad p_2 = \frac{\alpha(D-1)}{D}\\
		p_3 = \frac{\beta}{D}  &\qquad p_4 = \frac{\beta(D-1)}{D}\\
		p_5 = \frac{\gamma}{D} &\qquad p_6 = \frac{\gamma(D-1)}{D} 
	\end{split}
\end{equation*}

With this in hand, we know that the following condition implies that $S$ 
dominates $H$ when $\strategic$ has a block in hand:
$$
R_S = \sum_{i=1}^6 (S_i) p_i >  \sum_{i=1}^6 (H_i) p_i = R_H
$$

\subsection*{The Case where $k = 0$}
Just as in PPLNS, we focus on the scenario where the bag is empty and 
strategic pool miners find a block. In terms of the equations above, this amounts 
to the case where $k = 0$.

\begin{theorem}
	If $k = 0$, and $\alpha,\beta$ are such that 
	$$
	\alpha > \frac{1}{(D-1)^2}\left( \frac{ND}{N-1} + N - \beta(N-2) \right) = \Theta 
	\left( \frac{N+D}{D^2} \right)
	$$
	then $S$ strictly dominates $H$.
\end{theorem}

\begin{proof}
	
	We begin by computing the per-state surplus that $S$ gives to $H$, $\Delta_i 
	= S_i - H_i$, which gives us the following:
	
	\begin{align*}
		\Delta_1 &  = - \left( \frac{1}{N} + \frac{N-1}{N^2} + \frac{(N-1)^2}{N^2 D}  
		\right)  \\
		\Delta_2 &  = \frac{N-1}{N^2} \\
		\Delta_3 &  = - \left( \frac{1}{N} + \frac{N-1}{N^2} + \frac{N-1}{N^2 D} \right)   
		\\
		\Delta_4 &  = \frac{N-1}{N^2 D} \\
		\Delta_5 &  = - \left( \frac{1}{N} + \frac{N-1}{ND} \right) 
		\\
		\Delta_6 &  = 0 
	\end{align*}
	
	Given these expressions, we are interested in the scenarios where 
	$\sum_{i=1}^6 p_i \Delta_i > 0$. Given the fact that $N,D \geq 0$, this is 
	equivalent to $ND^2\sum_{i=1}^6 p_i \Delta_i > 0$. This leads to the following 
	condition on $\alpha$ for hoarding to dominate revealing a block.
	$$
	\alpha > \frac{1}{(D-1)^2}\left( \frac{ND}{N-1} + N - \beta(N-2) \right) = \Theta 
	\left( \frac{N+D}{D^2} \right)
	$$
\end{proof}

In summary, we have shown that as in PPLNS, if $\alpha = \Omega \left( 
\frac{N+D}{D^2} \right)$, then hoarding a block when up against an empty bag is 
strictly better than publishing it immediately. As in with PPLNS, it is often the 
case that $D = \Theta (N)$, hence this tells us that in this scenario, if $\alpha = 
\Omega  \left( \frac{1}{N} \right)$, then miners are strategic by hoarding blocks. 
This is indeed what our empirical recursion results show.

\section{Conclusion and Future Work}	
\label{sec:future-work}

In this paper we presented RPPLNS,  a novel twist on the already popular 
``Pay-per-last-$N$-shares'' (PPLNS) mining pool scheme used by the majority of 
the Bitcoin network. By suitably randomising PPLNS, we are able to maintain its 
strengths (fairness, variance reduction, pool-hop-proof-ness) while proving 
robustness guarantees for honest mining and reducing the underlying memory 
constraints of the protocol. There are several possible directions for future 
research on RPPLNS. In the following paragraphs, we touch upon the ones we 
consider the most fruitful.

\subsection{Interpolating between RPPLNS and PPLNS: Queue-bag 
protocols}
At its core, PPLNS and RPPLNS are quite similar. It is not difficult to see that for 
a given value of $N$, one can interpolate between PPLNS and RPPLNS by 
considering a pool protocol $\mathcal{M}_{Q,N}$ that maintains a queue of 
length $Q \in [N]$, and when a share is kicked out of the queue, it is put in a bag 
of size $N-Q$. 
%(As visualised in Figure \ref{fig:queue-bag}). 
Clearly $\mathcal{M}_{1,N}$ is RPPLNS and $\mathcal{M}_{N,N}$ is PPLNS. 
One can show that such family of queue-bag protocols share similar fairness, 
variance reduction and pool-hop-proof properties of PPLNS and RPPLNS. For 
such protocols, it would be interesting to study their incentive compatibility and 
see whether there is an optimal choice of $Q$ to be made in terms of variance 
reduction, incentive compatibility and memory usage minimisation.

\subsection{Possible Benefits of Informational Fairness}

We have already mentioned the fact that the reduced cardinality of the state 
space of RPPLNS vs. PPLNS can give rise to space usage gains for an 
implementation of the pool protocol. An interesting consequence of this fact is 
that since the state space of RPPLNS is  smaller, it is also succinctly 
describable, and hence easy to communicate to all pool miners. It would be 
interesting to know that if a pool is operating a PPLNS mining protocol and 
there are two strategic agents $m_1$ and $m_2$ within the pool, where $m_1$ 
has full knowledge of the state of the queue $s \in S$ at any given moment of 
time and $m_2$ only has partial information, say some statistic, $s'$ over $s$ 
(this could be the queue considered as a bag for example, with order of 
elements in the queue forgotten), whether this gives $m_1$ any undue advantage 
over $m_2$. If this were the case, there would be a strong further justification for 
RPPLNS, as a way of putting all agents on an equal informational playing 
ground.

\subsection{Stronger Incentive Guarantees}

It would be great to rigorously understand the recurrence relation governing the 
gain of an optimal strategic pool miner in RPPLNS. This would give us stronger 
results regarding incentive compatibility of pool miners and could glean insights 
into unforeseen strategic considerations.

%%
%% The acknowledgments section is defined using the "acks" environment
%% (and NOT an unnumbered section). This ensures the proper
%% identification of the section in the article metadata, and the
%% consistent spelling of the heading.
%\begin{acks}
%To Robert, for the bagels and explaining CMYK and color spaces.
%\end{acks}

%%
%% The next two lines define the bibliography style to be used, and
%% the bibliography file.
%\bibliographystyle{ACM-Reference-Format}
\bibliography{biblio}

%%
%% If your work has an appendix, this is the place to put it.
%\appendix

\end{document}